\documentclass[11pt]{article}

\usepackage{times}
\usepackage{algorithmic}
\usepackage{amsmath}
\usepackage{amssymb}
\usepackage{amsthm}
\usepackage{color}
\usepackage{colortbl}
\usepackage{float}
\usepackage{cite}
\usepackage{subfig}
\usepackage{verbatim}
\usepackage[pdftex]{graphicx}

\usepackage[top=.9in, bottom=.9in, left=.80in, right=.80in]{geometry}

\newtheorem{theorem}{Theorem}
\newtheorem{lemma}[theorem]{Lemma}

\newtheorem{example}[theorem]{Example}
\newtheorem{remark}[theorem]{Remark}

\newcommand{\deff}{\mbox{$\stackrel{\rm def}{=}$}}

\newcommand{\cB}{{\mathcal B}}
\newcommand{\cC}{{\mathcal C}}
\newcommand{\cG}{{\mathcal G}}

\newcommand{\cS}{{\mathcal S}}

\newcommand{\cP}{{\mathcal P}}

\newcommand{\cX}{{\mathcal X}}

\newcommand{\sP}{\cP}
\newcommand{\sG}{\cG}

\newcommand{\Gr}{\smash{{\sG\kern-1.5pt}_q\kern-0.5pt(n,k)}}
\newcommand{\Gk}{\smash{{\sG\kern-1.5pt}_q\kern-0.5pt(n,k_1)}}
\newcommand{\Gkk}{\smash{{\sG\kern-1.5pt}_q\kern-0.5pt(n,k_2)}}
\newcommand{\Grtwo}{\smash{{\sG\kern-1.5pt}_2\kern-0.5pt(n,k)}}
\newcommand{\Gkone}{\smash{{\sG\kern-1.5pt}_q\kern-0.5pt(n,k_1)}}
\newcommand{\Gktwo}{\smash{{\sG\kern-1.5pt}_q\kern-0.5pt(n,k_2)}}
\newcommand{\Ps}{\smash{{\sP\kern-2.0pt}_q\kern-0.5pt(n)}}

\begin{document}
\title{Optimal Combinatorial Batch Codes\\ based on Block Designs\let\thefootnote\relax\footnotetext{A preliminary version of the paper is available at \emph{http://arxiv.org/abs/1312.5505}}}

\author{Natalia Silberstein and Anna G\'al\thanks{N. Silberstein is with the Department of Computer Science,
  Technion --- Israel Institute of Technology, Israel,
  email: natalys@cs.technion.ac.il. Supported in part by a Fine Fellowship.
  A. G\'al is with the Department of Computer Science,
  University of Texas at Austin, Austin, Tx, USA
  email: panni@cs.utexas.edu.
Supported in part by NSF Grant CCF-1018060.
 }}

\maketitle

\begin{abstract}
Batch codes, introduced by Ishai, Kushilevitz, Ostrovsky and Sahai,
represent the distributed storage of an $n$-element data set on $m$ servers
in such a way that any batch of $k$ data items can be retrieved by reading at most one
(or more generally, $t$) items from each server, while keeping the total storage
over $m$ servers equal to~$N$.
This paper considers a class of batch codes (for $t=1$), called combinatorial batch codes (CBCs),
where each server stores a subset of a database.
A  CBC
is called optimal if the total storage $N$ is minimal for given $n,m$, and $k$.
A $c$-uniform CBC is a combinatorial batch code where each
item is stored in exactly $c$ servers. A $c$-uniform CBC is called optimal if its
parameter $n$ has maximum value for  given $m$ and $k$. Optimal $c$-uniform
CBCs have been known
only for $c\in \{2,k-1,k-2\}$.

In this paper we present new constructions of optimal CBCs
in both the uniform and general settings,
for values of the
parameters where tight bounds have not been established previously.
In the uniform setting, we provide  constructions of two new families of
optimal uniform codes with $c\sim \sqrt{k}$. Our constructions are based on
affine planes and transversal designs.
\end{abstract}


\section{Introduction}

%\pagenumbering{arabic}

Batch codes were
introduced by Ishai, Kushilevitz, Ostrovsky and Sahai~\cite{IKOS04}.
An $(n, N, k, m, t)$ \emph{batch code} over an alphabet $\Sigma$,
  encodes $x\in \Sigma^n$ into an $m$-tuple of strings
$y_1,\ldots, y_m\in \Sigma^*$ of total length $N$ (stored in $m$ servers),
such that for every $k$-tuple (batch) of indices $i_1,\ldots, i_k\in [n]$,
the $k$ data items $x_{i_1},\ldots, x_{i_k}$ can be retrieved by reading
at most $t$ symbols  from each server.
Batch codes were motivated by applications to load balancing in distributed
storage, private information retrieval and cryptographic protocols.
It is desirable to minimize the total storage $N$ used to store a data set of size $n$, or, equivalently, to maximize the rate of the code, defined as the ratio $n/ N$. Also, it is desirable to have the number of servers $m$ as small as possible, given the parameters $k,t$ and $n$.

{\bf Combinatorial Batch Codes:}
The name \emph{combinatorial batch codes} was proposed by
Paterson, Stinson, and Wei~\cite{PSW08} to refer to
purely replication based batch codes.
Combinatorial batch codes
is the class of batch codes,
where each server stores a subset of data items and decoding simply means reading items from servers.
An  $(n,N,k,m,t)$-CBC
is a combinatorial batch code storing $n$ data items
on $m$ servers with total storage size $N$, such that
any $k$ data items can be retrieved by reading
at most $t$ items from each server.
An  $(n,N,k,m,t)$-CBC
is called optimal if the total storage $N$ is minimal for given $n,m$, and $k$.
An $(n, N, k, m, t=1)$-CBC is denoted by $(n, N, k, m )$-CBC and the
minimal value of $N$ for $t=1$ is denoted by $N(n,k,m)$.

In this paper we consider only the case $t=1$. (Combinatorial batch codes for $t>1$ have been studied in \cite{BuTu12,RuRo08}.)
Note that when $t=1$,
we can assume that $n\geq m\geq k$.
As noted in \cite{IKOS04}, there are two trivial extreme solutions to the
problem for $t=1$:
replicating the data string $x$  in each server gives a code with $m=k$
(which is lowest possible since $m\geq k$ must hold), but the total storage
used is $kn$.
For the case $m=n$ every server simply stores one data item,
so the total storage used is $n$ which is optimal,
while the number of servers is very large.

It has been observed already in \cite{IKOS04} that combinatorial batch codes
for $t=1$ are equivalent to
(unbalanced) bipartite expander graphs with expansion factor 1.
Expander graphs have been extensively studied, but known probabilistic
and explicit constructions of bipartite expanders do not give
optimal combinatorial batch codes.
Note that $N \geq n$ has to hold (and this is tight when $m=n$),
and $N \leq kn - m(k-1)$ by a simple construction
which is optimal when $m=k$ \cite{PSW08}.
However, for  a certain range of parameters
optimal constructions and tight bounds on $N(n,k,m)$ are not
known even up to constant factors.

Precise values of $N(n,k,m)$ have been established for some special
settings of the parameters, and more generally when $n \geq \binom{m}{k-2}$.
For  fixed  $k\in\{2,3,4\}$, the values of $N(n,m,k)$ are presented
in~\cite{PSW08} and ~\cite{BuTu11Dual}.
The results  for general $k$ where precise values of
$N(n,k,m)$ and constructions of CBCs that achieve these bounds have been
established are summarized in the following table.

\vspace{.3cm}
\noindent
\begin{tabular}{|c|c|c|c|}
  \hline
  $n$  & $m$ & $N(n,k,m)$ & references\\ \hline \hline
  $n$  & $m=n$ & $n$ &~\cite{IKOS04,PSW08}\\\hline
  $n$ & $m=k$ & $kn-k(k-1)$& ~\cite{IKOS04,PSW08}\\\hline
  $n=m+1$  & $m$ & $m+k$& \cite{PSW08}\\\hline
  $n=m+2$  & $m$ & $\left\{\begin{array}{cc}
                        m+k-2+\lceil2\sqrt{k+1}\rceil & \textmd{if } m+1-k\geq \lceil\sqrt{k+1}\rceil \\
                        2m-2 +\lceil1+\frac{k+1}{m-k+1}\rceil & \textmd{if } m+1-k < \lceil\sqrt{k+1}\rceil
                      \end{array}\right.$&\cite{BKMS10,BuTu11Dual}
     \\\hline
 $n\geq (k-1)\binom{m}{k-1}$  & $m$ & $kn-(k-1)\binom{m}{k-1}$&\cite{PSW08} \\\hline
 $\binom{m}{k-2}\leq n\leq (k-1)\binom{m}{k-1}$ & $m$ & $n(k-1)-\left\lfloor\frac{(k-1)\binom{m}{k-1}-n}{m-k+1}\right\rfloor$&\cite{RuRo08, BRR11, BuTu11} \\\hline
  \hline
\end{tabular}

\vspace{.3cm}

The last bound was generalized in~\cite{BRR11} for any $1\leq n\leq (k-1)\binom{m}{k-1}$ as follows:
Let $1\leq s\leq k-1$ be  the least integer such that
\[
n\leq \frac{(k-1)\binom{m}{s}}{\binom{k-1}{s}}.
\]
Then
\begin{equation}
\label{eq:boundN}
N(n,k,m)\geq ns-\left\lfloor\frac{(k-s)\left(\frac{(k-1)\binom{m}{s}}{\binom{k-1}{s}}-n\right)}{m-k+1}\right\rfloor.
\end{equation}

Bound~(\ref{eq:boundN}) is attained by the construction given in~\cite{BRR11}
  for half of the values of $n$ in the range
$\binom{m}{k-2}-(m-k+1)A(m,4,k-3)\leq n\leq \binom{m}{k-2}$, where $A(m,4,k-3)$ is the maximum number
of codewords in a binary constant weight code of length $m$, weight $k-3$
and Hamming distance~$4$.
The question of the tightness of bound~(\ref{eq:boundN}) for $n<\binom{m}{k-2}-(m-k+1)A(m,4,k-3) $ remained open.

%%%%%%%%%%%%%%%%%%%%%%%%%%%%%%%%%%%%%%%%%%%%%%%%%%%%%%%%%%%%%%%%%%%%%%%%%%%%%%%%%%%%%%%%%%%%%%%%%%%%%%%%%%%%
{\bf Uniform Combinatorial Batch Codes:}
A $c$-\emph{uniform} $(n,cn,k,m)$-CBC is a combinatorial batch code where each
item is stored in exactly $c$ servers.  These codes were studied in~\cite{PSW08,BRR11,BaBh12,BuTu13}. The maximum value of $n$ for which there
exists a $c$-uniform $(n,cn,k,m)$-CBC is denoted by $n(m,c,k)$.
In other words, $n(m,c,k)$ is the maximum number of data items that can be stored using a uniform code with the given number $m$ of servers and the given $c$ and $k$. The following general upper bound on $n(m,c,k)$ was established in
~\cite{PSW08}:
\begin{equation}
\label{eq:UpBoundUniform}
n(m,c,k)\leq \frac{(k-1)\binom{m}{c}}{\binom{k-1}{c}}.
\end{equation}

It was shown in~\cite{PSW08} by probabilistic arguments that $n(m,c,k)=\Omega({m^{\frac{ck}{k-1}-1}})$ for fixed integers $k$
and $c$ (the constants in the $\Omega$ notation depend on $k$ and $c$).
Explicit constructions of uniform codes which attain bound~(\ref{eq:UpBoundUniform}) were given for $c\in\{k-1, k-2\}$ in ~\cite{PSW08}.
It was proved in~\cite{BaBh12} that
$n(m,c=2,k=5)=\left\lfloor\frac{m^2}{4}\right\rfloor$. A construction of codes based on complete bipartite graphs given in~\cite{PSW08} attains this bound.
Based on the connection between uniform combinatorial batch codes and the forbidden hypergraph problem~\cite{BES73} the following bounds for uniform codes were shown in~\cite{BaBh12}, however, these bounds only hold when $k$ and $c$
are fixed integers
(the constants hidden in the notation depend on $k$ and $c$):
\begin{itemize}
  \item $n(m,2,k)=O(m^{1+\frac{1}{\lfloor k/4\rfloor}})$, for $k\geq 4$;
  \item $n(m,2,k)=\Theta(m^{3/2})$, for $k=6,7,8$;
  \item $n(m,2,k)=\Theta(m^{4/3})$, for $k=9,10,11$;
  \item $n(m,2,k)=\Theta(m^{6/5})$, for $k=15,16,17$;
  \item $n(m,c,k)=o(m^c)$, for $k\geq 7$, $3\leq c\leq k-1-\lceil\log{k}\rceil$;
   \item $n(m,c,k)=\Theta(m^c)$, for $k\geq 6$, $k-\lceil\log{k}\rceil\leq c\leq k-1$.
\end{itemize}

The following bound was proved in~\cite{BuTu13}:
$$n(m,c,k)=O(m^{c-1+1/\left\lfloor\frac{k}{c+1}\right\rfloor}),$$
for $c\leq \frac{k}{2}-1$. This improves the general bound
(\ref{eq:UpBoundUniform}) when $k$ and $c$ are fixed integers,
but it is weaker than (\ref{eq:UpBoundUniform}) in the general case.
Note that optimal constructions of uniform CBCs for $k=3$ and $k=4$
are implied by the bounds of \cite{PSW08} when $c=2$.
However, bounds tight up to constant factors for the value of $n(m,c,k)$
are not known even for $c=2$ when  $k \geq 18$.

{\bf Our results:}
In this paper we give several constructions of optimal combinatorial batch
codes for settings of the parameters where tight bounds
were not known previously, even up to constant factors.

We answer
the question about the tightness of bound~(\ref{eq:boundN})
affirmatively,
for certain settings of the parameters in the range
$n<\binom{m}{k-2}-(m-k+1)A(m,4,k-3) $.
We construct a family of  CBCs that attain bound~(\ref{eq:boundN})  with $n<\binom{m}{k-2}-(m-k+1)A(m,4,k-3)$. More precisely, given a prime power $q$, we construct an optimal $(n,N,k,m)$-CBC with $n=q^2+q-1$, $N=q^3-q$, $k=q^2-q-1$, and $m=q^2-q$. This construction is based on a family of block designs,
called transversal designs.

Regarding uniform batch codes, we present two new families of optimal uniform
combinatorial batch codes which attain bound~(\ref{eq:UpBoundUniform})
for $k=c^2$ and $k=c^2+c+1$. Previously, optimal uniform CBCs were
known only for $c\in \{2,k-1,k-2\}$.
Our first
optimal uniform construction is based on affine planes
while the second one is based on transversal designs.

Our proofs are based on  the observation that the strong structural properties
of affine planes and transversal designs are well suited to obtain
CBCs with good parameters. In fact, we show that simply taking the incidence
matrix of affine planes yields optimal uniform CBCs.
Moreover, the proof of this result is fairly short and simple
(see Theorem \ref{trm:affine_plane}).
Transversal designs have similar structure to affine planes,
and we show that they can be modified to obtain optimal CBCs in both the
uniform and non-uniform settings.
However, to get tight results, the proofs based on transversal designs
are somewhat longer than the simple proof for affine planes.

To summarize our results, we note that
the following optimal $(n,N,k,m)$-CBCs are constructed in the paper:
\begin{itemize}
  \item non-uniform $(q^2+q-1,q^3-q,q^2-q-1,q^2-q)$-CBC;
  \item $q$-uniform $(q^2+q,q(q^2+q),q^2,q^2)$-CBC;
  \item $(q-1)$-uniform $(q^2-3,(q-1)(q^2-3),q^2-q-1,q^2-q-1)$-CBC,
\end{itemize}
where $q$ is a prime power.

Along the way, we also obtain the following asymptotically optimal
constructions:
\begin{itemize}
\item
$(q-1)$-uniform  $(q^2, q^3-q^2, q^2-q-1, q^2-q)$-CBC,\\
 such that the gap between the upper bound~(\ref{eq:UpBoundUniform})
and the number of data items $n$ is $q-2 = o(n)$;

\item
$(q-1)$-uniform $(q^2+q-3,(q-1)(q^2+q-3),q^2-q-1,q^2-q)$-CBC,\\
 such that the gap between the upper bound~(\ref{eq:UpBoundUniform})
and the number of data items $n$ is $1$,

\end{itemize}
where $q$ is a prime power.

The rest of the paper is organized  as follows. In Section~\ref{sec:preliminaries} we describe the combinatorial batch codes in terms of set systems and  dual set systems which satisfy Hall's condition, as proposed in~\cite{PSW08}, and define transversal designs and affine planes. In Section~\ref{sec:non-uniform} we present our construction for  new combinatorial batch codes from transversal designs and prove their optimality. In Section~\ref{sec:uniform}
we present our constructions for uniform combinatorial batch codes.

\section{Preliminaries}
\label{sec:preliminaries}

The equivalent definition of combinatorial batch codes in terms of set systems is given as follows~\cite{PSW08}.
An $(n, N, k,m,t)$ combinatorial batch code is a \emph{set system} $(X,\cB)$, where $X$ is a set of $n$ points (corresponding to items), $\cB$ is a collection of $m$ subsets (or blocks) of $X$ (corresponding to servers) and $N=\sum_{B\in \cB}|B|$, such that for each $k$-subset $\{x_{i_1},\ldots, x_{i_k}\}\subset X$ there exists a subset $C_i\subseteq B_i$, where $|C_i|\leq t$, for $1\leq i\leq m$, such that $\{x_{i_1},\ldots, x_{i_k}\}\subset\bigcup_{i=1}^mC_i$.
In the sequel, we will consider combinatorial batch codes with $t=1$, and
we refer to  such  codes as $(n, N, k, m)$-CBCs.

Given a set system $(X,\cB)$ with the points set
$X=\{x_{1},\ldots, x_{n}\}$ and the blocks set $\cB=\{B_1,\ldots,B_m\}$, its incidence matrix is a $m\times n$ binary matrix $A$, where
$$(A)_{i,j}=
\left\{\begin{array}{cc}
                1, & \textmd{ if } x_j\in B_i \\
                0, & \textmd{ if } x_j\notin B_i
              \end{array}\right.
$$

The incidence matrix $\Gamma$ of an $(n, N, k, m)$-CBC is defined as the $m\times n$ incidence matrix (with $N$ ones) of the corresponding set system.
The following lemma~\cite{PSW08} shows the properties of $\Gamma$.

\begin{lemma}
\label{lm:matrixCondition}
An $m\times n$ binary matrix $\Gamma$ with $N$ ones is an incidence matrix of an $(n,N, k,m)$-CBC if and only if  for any $k$ columns there is a $k\times k$ submatrix of $A$ which has at least one generalized diagonal containing $k$ ones.
\end{lemma}

It is useful to represent CBCs by the \emph{dual set system}, where the points correspond to servers and the blocks correspond to items~\cite{PSW08}. Each block (an item) in the dual system contains the points (servers) that store this particular item. In other words, let $\cC$ be an $(n, N, k, m)$-CBC with $n$ items  $x_1,\ldots, x_n$ and $m$ servers $s_1,\ldots, s_m$. $\cC$ is represented by a (dual) set system $(\cS,\cX)$, where $\cS=\{s_1,\ldots, s_m\}$ is the set of $m$ servers
and $\cX=\{X_1,\ldots,X_n\}$ is a collection of $n$ subsets (blocks) of $\cS$. If an item $x_j$, for $1\leq j\leq n$, is stored in servers $s_{i_1},\ldots s_{i_{\ell}}$ then $x_j$ is represented by a subset $X_j$, where $X_j=\{s_{i_1},\ldots s_{i_{\ell}}\}$. Note that it holds that $\sum_{X\in \cX}|X|=N$.

The necessary and sufficient condition for $\cC$ to be a CBC, in terms of dual systems, is given by using Hall's theorem~\cite{BRR11} and is presented  in the following lemma.

\begin{lemma}
\label{lm:matrixConditionAllr}
The necessary and sufficient condition that any set of $k$ items can be retrieved by reading at most one item per server is
 that given any $r$ sets  $X_{i_1},\ldots, X_{i_r}$ of $\cX$ , for all $r$, $1\leq r\leq k$,
in the dual system $(\cS,\cX)$,  it holds that $\cup_{1\leq j\leq r}X_{i_j}\geq r$.
In terms of an incidence matrix $\Gamma$ of a code, it means that
 for any set of $r$ columns of $\Gamma$,  $\{\Gamma_{i_1},\ldots, \Gamma_{i_r}\}$, $1\leq r\leq k$,  it holds that union of these columns (i.e., the characteristic vector of the union of the corresponding blocks of the dual systems) contains at least $r$ nonzero entries of~$\Gamma$.
\end{lemma}

Most of the constructions presented in this paper are based on a family of block designs, called transversal designs. The definition of this designs is as follows.

A \emph{transversal design} (TD) of group size $h$ and block size $\ell$,  denoted by $\text{TD}(\ell, h)$,
is a triple $(\mathcal{P},\mathcal{G},\mathcal{B})$, where
\begin{enumerate}
\item $\mathcal{P}$ is a set of $\ell h$ \emph{points};
\item $\mathcal{G}$ is a partition of $\mathcal{P}$ into $\ell$ sets
(\emph{groups}), each one of size $h$;
\item $\mathcal{B}$ is a collection of $\ell$-subsets of $\mathcal{P}$
(\emph{blocks});
\item each block meets each group in exactly one point;
\item any pair of points from different groups is contained in exactly one block.
\end{enumerate}
It follows from the definition of TD that the number of blocks in $\text{TD}(\ell, h)$ is $h^2$ and the number of blocks that contain a given point is $h$~\cite{Anderson}.
A $\text{TD}(\ell,h)$ is called \emph{resolvable} if the set $\mathcal{B}$
can be partitioned into sets $\mathcal{B}_1,...,\mathcal{B}_h$, each one containing $h$ blocks,
such that each element of $\mathcal{P}$ is contained in exactly one block of
each $\mathcal{B}_i$. The sets $\mathcal{B}_1,...,\mathcal{B}_s$ are called
\emph{parallel classes}.
The existence of  resolvable transversal designs is considered in the following theorem (see e.g. in~\cite{Anderson}).
\begin{theorem}
Let $q$ be a prime power. Then there exists a resolvable $\text{TD}(\ell,q)$ for any integer $\ell \leq q$.
\end{theorem}

\begin{example}
\label{ex:TD(3,4)}

We consider the resolvable transversal design $\textmd{TD}(3,4)$. The  points $\mathcal{P}=\{1,2,\ldots,12\}$, groups $\mathcal{G}=\{G_1,G_2,G_3\}$ and blocks $\mathcal{B}=\{B_1,B_2,\ldots, B_{16}\}$ with four parallel classes $\mathcal{B}_1,\mathcal{B}_2,\mathcal{B}_3,\mathcal{B}_4$, are given by

$$
\begin{array}{|c|c|c|}\hline
  G_1 & G_2 & G_3  \\\hline
  1 & 5 & 9 \\
  2 & 6 & 10\\
  3 & 7 & 11\\
  4 & 8 & 12\\\hline
\end{array}
%\;;
$$
$$
\begin{array}{|c|c|c|c|}\hline
  \mathcal{B}_1 & \mathcal{B}_2 & \mathcal{B}_3 & \mathcal{B}_4 \\\hline
  \begin{array}{c|c|c|c}
    B_1 & B_2 & B_3 & B_4 \\\hline
    1 & 2 & 3 & 4 \\
    5 & 6 & 7 & 8 \\
    9 & 10 & 11 & 12
  \end{array}
   & \begin{array}{c|c|c|c}
    B_5 & B_6 & B_7 & B_8 \\\hline
    1 & 2 & 3 & 4 \\
    6 & 5 & 8 & 7 \\
    11 & 12 & 9 & 10
  \end{array} & \begin{array}{c|c|c|c}
    B_9 & B_{10} & B_{11} & B_{12} \\\hline
    1 & 2 & 3 & 4 \\
    8 & 7 & 6 & 5 \\
    10 & 9 & 12 & 11
  \end{array} & \begin{array}{c|c|c|c}
    B_{13} & B_{14} & B_{15} & B_{16} \\\hline
    1 & 2 & 3 & 4 \\
    7 & 8 & 5 & 6 \\
    12 & 11 & 10 & 9
  \end{array}\\\hline
\end{array}
$$

The  transpose of the incidence matrix $A$ of $\text{TD}(3,4)$ is given by the following $12\times 16$ block matrix:
\end{example}
\begin{footnotesize}
$$A^T=\begin{array}{c}
     \left(\begin{tabular}{c|c|c|c}%\hline\hline
    1 0 0 0 & 1 0 0 0 & 1 0 0 0 & 1 0 0 0\\
    0 1 0 0 & 0 1 0 0 & 0 1 0 0 & 0 1 0 0\\
    0 0 1 0 & 0 0 1 0 & 0 0 1 0 & 0 0 1 0\\
    0 0 0 1 & 0 0 0 1 & 0 0 0 1 & 0 0 0 1\\
     \hline
    1 0 0 0 & 0 1 0 0 & 0 0 0 1 & 0 0 1 0\\
    0 1 0 0 & 1 0 0 0 & 0 0 1 0 & 0 0 0 1\\
    0 0 1 0 & 0 0 0 1 & 0 1 0 0 & 1 0 0 0\\
    0 0 0 1 & 0 0 1 0 & 1 0 0 0 & 0 1 0 0\\
     \hline
    1 0 0 0 & 0 0 1 0 & 0 1 0 0 & 0 0 0 1\\
    0 1 0 0 & 0 0 0 1 & 1 0 0 0 & 0 0 1 0\\
    0 0 1 0 & 1 0 0 0 & 0 0 0 1 & 0 1 0 0\\
    0 0 0 1 & 0 1 0 0 & 0 0 1 0 & 1 0 0 0\\
         \end{tabular}
\right)
\end{array}$$
\end{footnotesize}

%==============================================================================================================
A construction of an optimal uniform code presented in Section~\ref{sec:uniform} is based on an affine plane. The definition of an affine plane, in terms of set systems, is given as follows:

 An \emph{affine plane} of order $s$, denoted by $A(s)$, is  a set system $(X,\mathcal{B})$, where $X$ is a set of $|X|=s^2$ points, $\mathcal{B}$ is a collection of $s$-subsets (blocks) of $X$ of size $|\mathcal{B}|=s(s+1)$, such that each pair of points in $X$ occur together in exactly one
block of~$\mathcal{B}$.
 It follows from the definition, that an affine plane is always \emph{resolvable}, that is the set $\cB$ can be partitioned into $s+1$ sets of size $s$, called parallel classes, such that every element of $X$ is contained in exactly one block of each class.
 The existence of affine planes is considered in the following theorem (see e.g. in~\cite{Anderson}).
\begin{theorem} If $q$ is a prime power, then there exists an affine plane of order $q$.
\end{theorem}

\begin{remark}
Note that an affine plane of order $q$ is equivalent to $\text{TD}(q+1,q)$: the transpose of an incidence matrix of an affine plane of order $q$ is the incidence matrix of $\text{TD}(q+1,q)$.
\end{remark}

\section{Construction of Optimal CBCs from Transversal Designs}
\label{sec:non-uniform}

In this section we present a construction of new optimal combinatorial batch codes.  These new batch codes prove the tightness of bound~(\ref{eq:boundN}) for a new range of parameters. The construction makes use of resolvable transversal designs.

\textbf{Construction I:}
Let $q\geq 3$ be a prime power. Let $ \text{TD}(q)\deff \text{TD}(q-1,q)$ be a resolvable transversal design with block size $q-1$ and group size $q$.
We define the servers of a code to be the  points of $\text{TD}(q)$ and the items of the code to be the union of blocks and groups of $\text{TD}(q)$. We denote the CBC constructed from $\text{TD}(q)$ by $\cC_{\text{TD}}(q)$.

\begin{example}
The  incidence matrix $\Gamma$ of  the $\cC_{\text{TD}}(4)$ obtained from $\text{TD}(3,4)$ from Example~\ref{ex:TD(3,4)} is given by
\end{example}
\begin{footnotesize}
$$\Gamma=\begin{array}{c}
     \left(\begin{tabular}{c|c|c|c|c}%\hline\hline
    1 0 0 0 & 1 0 0 0 & 1 0 0 0 & 1 0 0 0 & 1 0 0\\
    0 1 0 0 & 0 1 0 0 & 0 1 0 0 & 0 1 0 0 & 1 0 0\\
    0 0 1 0 & 0 0 1 0 & 0 0 1 0 & 0 0 1 0 & 1 0 0\\
    0 0 0 1 & 0 0 0 1 & 0 0 0 1 & 0 0 0 1 & 1 0 0\\
     \hline
    1 0 0 0 & 0 1 0 0 & 0 0 0 1 & 0 0 1 0 & 0 1 0\\
    0 1 0 0 & 1 0 0 0 & 0 0 1 0 & 0 0 0 1 & 0 1 0\\
    0 0 1 0 & 0 0 0 1 & 0 1 0 0 & 1 0 0 0 & 0 1 0\\
    0 0 0 1 & 0 0 1 0 & 1 0 0 0 & 0 1 0 0 & 0 1 0\\
     \hline
    1 0 0 0 & 0 0 1 0 & 0 1 0 0 & 0 0 0 1 & 0 0 1\\
    0 1 0 0 & 0 0 0 1 & 1 0 0 0 & 0 0 1 0 & 0 0 1\\
    0 0 1 0 & 1 0 0 0 & 0 0 0 1 & 0 1 0 0 & 0 0 1\\
    0 0 0 1 & 0 1 0 0 & 0 0 1 0 & 1 0 0 0 & 0 0 1\\
         \end{tabular}
\right)
\end{array}$$
\end{footnotesize}

Before we analyze the parameters of the constructed CBC, we present the properties of the incidence matrix $\Gamma$ of $\cC_{\textmd{TD}}(q)$ which
will be useful in the proofs.

Let $A$ be a $q^2\times q(q-1)$ incidence matrix of $\text{TD}(q)$, where rows of $A$ correspond to the blocks of $\text{TD}(q)$,
and columns of $A$ correspond to the points of $\text{TD}(q)$. Since $\text{TD}(q)$ is resolvable and all its groups are disjoint by definition, there is a permutation of rows and columns of $A$ which results in
a matrix that consists of $q(q-1)$ permutation matrices, each of size
$q \times q$.
Each  $q\times q$ permutation matrix corresponds to the $q$ points of a group of $\text{TD}(q)$ and $q$ blocks of a parallel class of $\text{TD}(q)$. From now on we assume that $A$ has this form.

Let $G$ be a $(q-1)\times q(q-1)$ matrix where the rows are
the incidence vectors of groups of $\text{TD}(q)$, i.e.,
the $i$th row of $G$ is a binary vector with $q$ 1s in
positions $(i-1)q+j$, $1\leq j\leq q$.
Denote by $\Gamma=(A^T||G^T)$ the $q(q-1)\times (q^2+q-1)$ matrix,
where the first $q^2$ columns are formed by the columns of $A^T$
(incidence vectors of blocks of $\text{TD}(q)$), and the last $q-1$ columns are
formed by the columns of $G^T$ (incidence vectors of groups of $\text{TD}(q)$).

Note that $\Gamma$ has the following structure:
its $q^2 + q - 1$ columns can be partitioned into $q+1$ classes,
where the first $q$ classes contain $q$ columns each, and correspond to the
parallel classes of $\text{TD}(q)$, and the last class contains $q-1$
columns (the incidence vectors of the groups of $\text{TD}(q)$).
We refer to the first $q$ classes as the parallel classes, and the last class
as the {\em special} class.
If not specified, a class of columns can be either one of the parallel classes
or the special class.
Note also that the first $q^2$ columns contain $q-1$ 1s each,
and the last $q-1$  columns (of the special class) contain $q$ 1s each.

We will say that ``\emph{a column $\Gamma_i$ of $\Gamma$ covers
the set $S$ of points}'' if the
block or the group of $\text{TD}(q)$ corresponding  to $\Gamma_i$
contains all the points of $S$.

The following simple observations will be used in the sequel.

\begin{itemize}

\item (A) The columns of $\Gamma$ within a given parallel class  are disjoint,
thus $\ell$ columns of a parallel class cover $\ell (q-1)$ points.

\item (B) Any two columns of $\Gamma$ from different parallel classes intersect in
at most one common point, thus $\ell_1$ columns from a parallel class
together with $\ell_2$ columns from another parallel class cover
at least $\ell_1(q-1) + \ell_2(q-1-\ell_1)$ points.
Note that this is useful when $\ell_1 \leq q-2$.

\item (C) Any column of $\Gamma$ from the special class covers all $q$ points of
one of the groups. On the other hand, $\ell$ columns of a parallel class
cover only $\ell$ points from each group.
Thus, $\ell$ columns of a parallel class together with $x$ columns
of the special class cover at least
$\ell (q-1) + x (q-\ell)$ points.

\end{itemize}

We will also use the following lemmas.

\begin{lemma}
\label{lm:permutation}
Let $P_1$ and $P_2$ be two parallel classes of $\text{TD}(q)$.
There is at most one column (block) from $P_2$  whose points are
covered by  $q-1$ columns of $P_1$.
\end{lemma}
\begin{proof}
Suppose that $q-1$ blocks of $P_1$ cover two blocks of $P_2$, denoted
by $a$ and $b$.
Let $c$ be the remaining block of $P_1$. Since $c$ has no common
points with $a$ and $b$, $c$ should be covered by the $q-2$ blocks
in $P_2\setminus\{a,b\}$. However, since $c$ contains $q-1$ points,
at least one of the $q-2$ blocks in $P_2\setminus\{a,b\}$ must intersect
 $c$ in at least two points,
which contradicts Property (5) of transversal designs.
\end{proof}
\begin{lemma} \label{td2q}
In $\textmd{TD}(2,q)$, given $2(q-2)$ points covered by $q-2$ blocks from each of two different parallel classes,
there is no further parallel class having $q-2$ blocks that cover these points.
\end{lemma}
\begin{proof}
Suppose there are three parallel classes such that $q-2$
blocks from each cover the same $2(q-2)$ points of $\textmd{TD}(2,q)$.
Consider the set of remaining 4 points, denoted by $S_4$.
First note that the 4 points in $S_4$ must be covered by the remaining
two blocks of each of the above three classes.
Next, note that since $S_4$ contains exactly two points from each group,
there are 4 possible blocks that can be formed by the points of $S_4$.
However, there are six blocks (two from each of the three classes) that
must be formed using these 4 points.
Therefore, at least one block must appear at least twice,
which contradicts Property (5) of transversal designs.
\end{proof}

\begin{lemma}
\label{lm:two_colomns}
Let $P_1$ be a parallel class of $\textmd{TD}(y,q)$, $q\geq y\geq 3$.
Then any $2$ blocks of any other parallel class intersect with
any given $2$ blocks in the class $P_1$ in at most $4$ points,
and thus cover $2y-4$ additional  points.
\end{lemma}
\begin{proof}
The proof directly follows from  Property (5)  of transversal
 designs.
\end{proof}

\begin{lemma}\label{lm:25}
In $\textmd{TD}(q)$ let $P=\{p_1,p_2,p_3\}$
be three points which  are contained in three different groups.
Then, for any set of blocks $R$, there are at most three parallel classes
that contribute $q-2$ blocks each to the set $R$ such that none
of these blocks contain points from $P$.
\end{lemma}

\begin{proof}
Suppose that for some set of blocks $R$, there are four
parallel classes $P_1,P_2,P_3,P_4$ that
contribute $q-2$ blocks each to $R$, such that none of these blocks contain
points from $P$.
Consider the remaining two blocks from each  $P_i$, $1\leq i\leq 4$.
For each $P_i$, $1\leq i\leq 4$,  at least one of these two blocks should
contain at least two points from $P$. Since there are only three different
pairs of points in $P$, there exists at least one pair of points in $P$
that is contained in at least two blocks.
This contradicts Property (5) of transversal designs.
\end{proof}

Now we have all the machinery needed to prove the following theorem.

\begin{theorem}
\label{trm:td-cbc}
The code $\cC_{\text{TD}}(q)$ obtained from $\text{TD}(q)$ is a $(q^2+q-1, q^3-q, q^2-q-1, q^2-q)$-CBC.
\end{theorem}
\begin{proof} First, since the number of items for the code is equal to
the number of blocks
plus the number of
groups of $\text{TD}(q)$, it follows  from the definition of TD that $n=q^2+q-1$.
Second, the number of servers is equal to the number of points of $\text{TD}(q)$, and then $m=q(q-1)$.
Since every point in $\text{TD}(q)$ is contained in $q$ blocks and one group, we have $N=q(q-1)(q+1)=q^3-q$.

To prove that  $k=q^2-q-1$, by Lemma~\ref{lm:matrixConditionAllr} we need to show that
\begin{itemize}
\item There exists a set of $q^2-q$ blocks and groups of $\text{TD}(q)$, such that their union contains at most $q^2-q-1$ points, in other words, $k\leq q^2-q-1$;
  \item For any set of $r$ blocks and groups of $\text{TD}(q)$, $1\leq r\leq q^2-q-1$, their union contains at least $r$ points, in other words, $k\geq q^2-q-1$.
  \end{itemize}

Let $p$ be a point of $\text{TD}(q)$. Since $\text{TD}(q)$ is a resolvable transversal design,  there are $q$ parallel classes, each one of size $q$, which partition the set of blocks of $\text{TD}(q)$. From each parallel class of blocks of $\text{TD}(q)$ we take $q-1$ blocks (all the blocks except one) which does not contain $p$. We obtained  $q(q-1)$ different blocks, such that their union does not contain the point $p$, in other words,  their union contains at most $q^2-q-1$ points. Then $k\leq q^2-q-1$.

To show that $k\geq q^2-q-1$, we will prove that any set of $r$ columns of $\Gamma$, $1\leq r\leq q^2-q-1$,  covers at least $r$
points.
Let $R$ be an arbitrary set of $r$ columns of $\Gamma$ with
$r=s + x$, where $s$ is the number of columns of $R$ from the parallel classes, and $x$, $0 \leq x \leq q-1$, is  the number of columns of $R$ which belong to the special class.
We use the notation $s=iq+j$, where $0\leq i\leq q-2$  and
$0\leq j\leq q-1$.
Let $t$ be the maximum number of columns  which is contributed to $R$ by a parallel class. Note that $t\geq i$. We consider the following cases:

\textbf{Case $t\geq i+2$}.
First, if $i=q-2$ then by (A) $i+2=q$ columns of
a parallel class cover $q(q-1)> r$ points.
Then we assume that $i\leq q-3$. By (C), $i+2$ columns from a parallel class
with $x$ columns of the special class cover
at least $(i+2)(q-1)+x(q-i-2)=iq+q+(x+1)(q-i-2)\geq s+x $ points.

\textbf{Case $t=i+1$}. In addition to the parallel class which contributes $i+1$ columns to $R$ there exists at least one parallel class which contributes at least $i$ columns. By
(B), blocks from these two parallel classes cover at least $(i+1)(q-1)+i(q-i-2)=iq+q-1+i(q-i-3)$ points. Then for $i\leq q-4$ and $x\leq i$ we have enough covered points.

On the other hand, $i+1$ columns from one parallel class together with $x$ columns of the special class by (C) cover a least $(i+1)(q-1)+x(q-i-1)=iq+q-1+x(q-i-1)-i$ points. Then for $i\leq q-3$ and $x\geq i$ we have enough covered points.
So we need to consider the following sub-cases:
\begin{itemize}
  \item $i=q-2$. If there are $x=q-1$ columns of the spacial class in $R$, then these columns cover $q(q-1)> r$ points.
  %If $x=q-2$ then by (3) the columns from the special class together with $t$ columns  of a parallel class cover at least $(q-1)^2+q-2=q^2-q-1\geq r$ points.
  Now assume that $x\leq q-2$. If there are two parallel classes that contribute $q-1$ columns each to $R$, then these columns together cover at least $(q-1)^2+q-2=q^2-q-1\geq r$ points, by Lemma~\ref{lm:permutation}. If there is only one such parallel class
   then the number of selected columns from parallel classes is at most $(q-1)+(q-2)(q-1)=(q-2)q+1$, hence
   $j\leq 1$, and thus $x$ columns from the special class together with $t$ columns  of a parallel class cover at least $(q-1)^2+x=(q-2)q+1+x\geq r$ points, by (C).

  \item $i=q-3$ and $x\leq q-4$. If there are two parallel classes that contribute $t=q-2$ columns each to $R$, then these columns together cover at least $(q-2)(q-1)+2q-6=(q-3)q+2q-4\geq r$ points, by Lemma~\ref{lm:two_colomns}. If there is only one such parallel class then
      the number of selected columns from parallel classes is at most $(q-2)+(q-3)(q-1)=(q-3)q+1$, hence
   $j\leq 1$. Thus $q-2$ columns from the parallel class which contributes $t$ columns to $R$ together with $q-3$ columns from another parallel class cover at least $(q-2)(q-1)+q-3\geq(q-3)q+1+(q-4)\geq s+x$ points, by (B).
\end{itemize}

\textbf{Case $t=i$}. In this case each parallel class contributes exactly $i$ columns, and hence $s=iq$. Any two parallel classes cover at least $i(q-1)+i(q-i-1)=iq+i(q-i-2)$ points, by (B). Then for $x=0$ or  $i\leq q-3$ and $x\leq i$ we have enough covered points. On the other hand,
$i$ columns from one parallel class together with $x$ columns of the special class by (C) cover a least $i(q-1)+x(q-i)=iq+x+x(q-i-1)-i$ points. Then for $i\leq q-2$ and $x\geq i$ we have enough covered points.

So we consider the only remaining sub-case, $i=q-2$ and $1\leq x\leq q-3$. Note that in this case $q\geq 4$, otherwise $x=0$. Therefore, there are at least 4 parallel classes that contribute $q-2$ columns each to $R$.
First we consider two parallel classes that contributes $q-2$ columns each. By Lemma~\ref{lm:two_colomns}, these columns cover at least $(q-2)(q-1)+2(q-1)-4=(q-2)q+q-4$ points. If $x\leq q-4$, then we are done.
Let $x=q-3$. To prove that the columns from two additional parallel classes cover at least one additional point,  we note that $2(q-2)$ columns cover all but $4$ points.
Suppose first that these 4 points are contained in two groups.
Since the columns from the first class cover $q-2$ points in each group,
we have two of the 4 points in both groups.
Then Lemma~\ref{td2q} implies that the remaining parallel classes must
cover at least one of the 4 points, and thus they cover at least one
additional point.
Next consider if these 4 points are distributed between
at least three groups. Since at least $4$ parallel classes contribute $q-2$ columns to $R$,
 Lemma~\ref{lm:25} implies that we do get one additional point from the second
two parallel classes.
Altogether, the columns contributed by the 4 parallel classes cover at
least $(q-2)q+q-4+1= r$ points.

\end{proof}

The following theorem proves the optimality of  $\cC_{\text{TD}}(q)$.

\begin{theorem}
\label{trm:td-cbc-Optimality}
The code $\cC_{\text{TD}}(q)$ is an optimal CBC attaining bound~(\ref{eq:boundN}) with $s=q$.
\end{theorem}

\begin{proof}

First, we prove that the smallest integer $1\leq s\leq k-1$, such that
\begin{equation}
\label{eq:uniformBound}
n\leq \frac{(k-1)\binom{m}{s}}{\binom{k-1}{s}},
\end{equation}
where $n=q^2+q-1$, $k= q^2-q-1$, $m= q^2-q$,
is  $s=q$.
We write~(\ref{eq:uniformBound}) as a function of $q$:
\begin{equation}\label{eq:c-bound}
q^2+q-1\leq \frac{(q^2-q-2)\binom{q^2-q}{s}}{\binom{q^2-q-2}{s}}.
\end{equation}

Note that the  function
$U_{m,k,s}=\frac{(k-1)\binom{m}{s}}{\binom{k-1}{s}}$ is an
increasing function of $s$, for fixed $m$ and $k$. One can easily verify
that for $s=q$ the inequality~(\ref{eq:c-bound}) holds, while for all $s<q$, this inequality does not hold.

Next, we  show that $\cC_{\text{TD}}(q)$ attains bound~(\ref{eq:boundN}).
Note, that $n\leq (k-1)\binom{ m}{k-1}=\frac{m(m-1)(m-2)}{2}$.
 We will prove that  for $s=q$ it holds that
 $$N=ns-\left\lfloor\frac{(k-s)(U_{m,k,s}-n)}{m-k+1}\right\rfloor.$$
  We express the values of $m$, $k$, $s$ as functions of $q$  and obtain
 % \[
%nc-\left\lfloor\frac{(k-c)(U_{m,k,c}-n)}{m-k+1}\right\rfloor
%=(q^2+q-1)q-\left\lfloor\frac{2q^3-4q^2}{2(q-2)}\right\rfloor=(q^2+q-1)q-q^2=q^3-q=N.
%\]
\begin{align*}
&ns-\left\lfloor\frac{(k-s)(U_{m,k,s}-n)}{m-k+1}\right\rfloor \\
&=(q^2+q-1)q-\left\lfloor\frac{2q^3-4q^2}{2(q-2)}\right\rfloor \\
&=(q^2+q-1)q-q^2\\
&=q^3-q\\
&=N.
\end{align*}
\end{proof}

The following theorem establishes  that $\cC_{\text{TD}}(q)$ shows the tightness of bound~(\ref{eq:boundN}) for a new range of parameters.
\begin{theorem}
\label{thm:new_range}
The parameters of $\cC_{\text{TD}}(q)$ satisfy $n\leq \binom{m}{k-2}-(m-k+1)A(m,4,k-3)$, where $A(m, 4, k-3)$ is the maximum cardinality of a constant weight code of length $m$, distance $4$ and constant weight $k-3$.
\end{theorem}

\begin{proof}
First note that by~\cite{AVZ00},

\begin{align*}
&A(m,4,m-4)=A(m,4,4)\\
&\leq \left\lfloor\frac{m}{4}A(m-1,4,3)\right\rfloor\\
&\leq \left\lfloor\frac{m}{4}\left\lfloor\frac{m-1}{3}A(m-2,4,2)\right\rfloor\right\rfloor\\
&\leq \left\lfloor\frac{m}{4}\left\lfloor\frac{m-1}{3}\left\lfloor\frac{m-2}{2}\right\rfloor\right\rfloor\right\rfloor.
\end{align*}
Then for $k=m-1$ we have
\begin{align}
\label{eq:paramTD}
&\binom{m}{k-2}-(m-k+1)A(m,4,k-3) \notag\\
&=\binom{m}{m-3}-2A(m,4,m-4) \notag\\
&\geq \binom{m}{m-3}-2\left\lfloor\frac{m}{4}\left\lfloor\frac{m-1}{3}\left\lfloor\frac{m-2}{2}\right\rfloor\right\rfloor\right\rfloor.
\end{align}
For $m=q^2-q$ equation~(\ref{eq:paramTD}) is greater than or equal to
$\frac{1}{12}(q^2-q)(q^2-q-1)(q^2-q-2)
$,
 which is larger than $n$ for $q\geq 4$. For $q=3$ we have $A(6,4,2)=3$, then $n=3^2+3-1=11\leq \binom{6}{3}-2\cdot3=14$.

\end{proof}

\section{Constructions of Uniform Combinatorial Batch Codes}
\label{sec:uniform}
In this section we present the constructions of two families of optimal $c$-uniform batch codes, both with $c\sim\sqrt{k}$. The first family of codes is based on affine planes, and the second one is based on transversal designs.

\subsection{Optimal Uniform Combinatorial Batch Codes from Affine Planes}

We present a family of optimal uniform batch codes
attaining the  bound~(\ref{eq:UpBoundUniform}) with $c=\sqrt{k}$.
This construction is based on the incidence matrix of affine planes.

\textbf{Construction II:} Let $q\geq 3$ be a prime power. Let $A(q)$ be an affine plane of order $q$. We define the servers of the code to be the
$q^2$ points of $A(q)$ and the items of the code to be the $q(q+1)$ blocks of $A(q)$. We
denote the uniform CBC constructed from $A(q)$ by $\cC_{A}(q)$.

\begin{remark} Note, that the incidence matrix of $\cC_A(q)$ is equal to the transpose of the incidence matrix of $A(q)$.
\end{remark}

\begin{theorem}
\label{trm:affine_plane}
 $\mathcal{C}_{A(q)}$ is a $q$-uniform $(q^2+q, q^3+q^2, q^2, q^2)$-CBC.
\end{theorem}

\begin{proof}
The parameters $n, N, m, c$ directly follow from the parameters of $A(q)$.

We will prove that $k=q^2$. Obviously, $k\leq q^2$.
We consider a set $R$ of  $r$ blocks, $1\leq r\leq q^2$.
First, if $r=q^2=(q+1)(q-1)+1$, then from the resolvability of $A(q)$, there is a parallel class that contributes $q$ blocks to $R$. Then, these $q$ blocks cover $q^2=r$ points.
Second, we assume that $r=i(q+1)+j=iq+i+j$, for $0\leq i\leq q-2$, $1\leq j\leq q+1$.
If there is a class that contributes $i+2$ blocks, then these blocks cover $(i+2)q=iq+2q>iq+i+j=r$ points. Now we assume that every parallel class contributes  at most $i+1$ blocks to $R$. More precisely, there are at least $j$ classes that contribute $i+1$ blocks and at most $q+1-j$ classes that contribute at most $i$ blocks to $R$. We consider the following cases:
\begin{itemize}
   \item $j=1$. Then, there is a class which contributes $i+1$ blocks and a class which contributes $i$ blocks. Similarly to observation (B) on transversal designs, these blocks cover $(i+1)q+(q-i-1)i\geq iq+q+i>r$ points.
  \item $j\geq 2$. In this case, there are at least two classes which contribute $i+1$ blocks to $R$. Similarly to observation (B) on transversal designs, these blocks cover $(i+1)q+(q-i-1)(i+1)\geq iq+q+i+1=i(q+1)+(q+1)\geq i(q+1)+j=r$ points.
\end{itemize}
\end{proof}

\begin{theorem}
The code $\cC_A(q)$ is an optimal $q$-uniform CBC attaining
bound~(\ref{eq:UpBoundUniform}).
\end{theorem}
\begin{proof}
It holds that
$$\frac{(k-1)\binom{m}{c}}{\binom{k-1}{c}}=\frac{(q^2-1)\binom{q^2}{q}}{\binom{q^2-1}{q}}=\frac{(q^2-1)q^2}{q^2-q}=q(q+1)=n.
$$
\end{proof}

%%%%%%%%%%%%%%%%%%%%%%%%%%%%%%%%%%%%%%%%%%%%%%%%%%%%%%%%%%%%%%%%%%%%%%%

\subsection{Optimal Uniform Combinatorial Batch Codes from Transversal Designs}
We  first present  two constructions of asymptotically
optimal uniform codes
based on transversal designs, with $m=k+1$ and $k=c^2+c$, i.e., $c=\frac{\sqrt{4k+1}-1}{2}$
and then  modify these constructions to obtain optimal uniform codes
with $m=k$ and $k = c^2 + c + 1$, i.e., $c=\frac{\sqrt{4k-3}-1}{2}$.

\textbf{Construction III:}
Let $q\geq 3$ be a prime power and let $\text{TD}(q)$ be a resolvable transversal design $\text{TD}(q-1,q)$, as in Section~\ref{sec:non-uniform}. We define
the servers of the code to be the  $q(q-1)$ points of $\text{TD}(q)$  and the set of items of the code to be the $q^2$ blocks of $\text{TD}(q)$. We denote the uniform CBC constructed from $\text{TD}(q)$ by $\cC_1(q)$.

\begin{theorem}
\label{cor:uniform codes}
The code $\cC_1(q)$ is a $(q-1)$-uniform  $(q^2, q^3-q^2, q^2-q-1, q^2-q)$-CBC, such that the gap between the upper bound~(\ref{eq:UpBoundUniform}) and  the number of data items of $\cC_1$ is equal to $q-2$.
 \end{theorem}

 \begin{proof} The parameters of $\cC_1$ directly follow from Theorem~\ref{trm:td-cbc}, since the incidence matrix of $\cC_1$ is the submatrix of the incidence matrix of the code $\cC_{\textmd{TD}}(q)$ from Theorem~\ref{trm:td-cbc}.

 Since the  the number of data items of a code is an integer number, then we can rewrite bound~(\ref{eq:UpBoundUniform}) as
  $$n(m,c,k)\leq \left\lfloor\frac{(k-1)\binom{m}{c}}{\binom{k-1}{c}}\right\rfloor.
  $$
 Now, given that $m=q^2-q, k=q^2-q-1$, and $c=q-1$ we have

 \[\left\lfloor\frac{(k-1)\binom{m}{c}}{\binom{k-1}{c}}\right\rfloor-n
 =\left\lfloor\frac{(q^2-q-2)\binom{q^2-q}{q-1}}{\binom{q^2-q-2}{q-1}}\right\rfloor-q^2
 =\left\lfloor(q^2+q-2)+\frac{q^2-5q+6}{q^2-3q+2}\right\rfloor-q^2
 =q-2
 \]
 \end{proof}

Next, we modify Construction III  to obtain a uniform batch code $\cC_2$ such that the gap  between bound~(\ref{eq:UpBoundUniform}) and  the number of data items of $\cC_2$ is equal to $1$. We present the construction in terms of the incidence
matrix for the code.

\textbf{Construction IV:}
Let the columns of the matrix $\Gamma$ be the union of the columns
of the incidence matrix of $\cC_1(q)$  and
 $q-3$ columns of weight $q-1$, where the $q-1$ \emph{ones} of a new column $i$, $1\leq i\leq q-3$, are in positions $(i-1)q+j$, $2\leq j\leq q$.
 We denote the uniform CBC with incidence matrix $\Gamma$
by $\cC_2(q)$.

\begin{example} The  incidence matrix of the uniform code $\cC_2(4)$ obtained from Construction IV with $q=4$ is given by
\end{example}

\begin{footnotesize}
$$\Gamma=\begin{array}{c}
     \left(\begin{tabular}{c|c|c|c|c}%\hline\hline
    1 0 0 0 & 1 0 0 0 & 1 0 0 0 & 1 0 0 0 & 0\\
    0 1 0 0 & 0 1 0 0 & 0 1 0 0 & 0 1 0 0 & 1\\
    0 0 1 0 & 0 0 1 0 & 0 0 1 0 & 0 0 1 0 & 1\\
    0 0 0 1 & 0 0 0 1 & 0 0 0 1 & 0 0 0 1 & 1\\
     \hline
    1 0 0 0 & 0 1 0 0 & 0 0 0 1 & 0 0 1 0 &  0\\
    0 1 0 0 & 1 0 0 0 & 0 0 1 0 & 0 0 0 1 &  0\\
    0 0 1 0 & 0 0 0 1 & 0 1 0 0 & 1 0 0 0 &  0\\
    0 0 0 1 & 0 0 1 0 & 1 0 0 0 & 0 1 0 0 &  0\\
     \hline
    1 0 0 0 & 0 0 1 0 & 0 1 0 0 & 0 0 0 1 & 0 \\
    0 1 0 0 & 0 0 0 1 & 1 0 0 0 & 0 0 1 0 & 0 \\
    0 0 1 0 & 1 0 0 0 & 0 0 0 1 & 0 1 0 0 & 0 \\
    0 0 0 1 & 0 1 0 0 & 0 0 1 0 & 1 0 0 0 & 0 \\
         \end{tabular}
\right)
\end{array}$$
\end{footnotesize}

  \begin{theorem}
 \label{lm:almost optimal uniform}
  The code $\cC_2(q)$ is a $(q-1)$-uniform  $(q^2+q-3, (q-1)(q^2+q-3), q^2-q-1, q^2-q)$-CBC, which almost attains bound~(\ref{eq:UpBoundUniform}), i.e. the difference between the bound and  the number of data items of the code is 1.
 \end{theorem}

 \begin{proof}
 First, the parameters $n$, $N$, $m$, and  $c$ directly follow from the definition of the code.
 Hence, we only need to prove that the parameter $k$ of this construction is  the same as in Construction III, i.e. $k=q(q-1)-1$.
 The proof is similar to the proof of Theorem~\ref{trm:td-cbc}, and we use the same notations:  Let $\Gamma$ be the  incidence matrix of $\cC_2(q)$, where the first $q^2$ columns are partitioned into $q$ parts corresponding to parallel classes of $\textmd{TD}(q)$ and the last $q-3$ columns are the additional columns of weight $q-1$ each, which we will refer as the \emph{special} class.
  As in the proof of Theorem~\ref{trm:td-cbc}, we consider an arbitrary set $R$ of $r$ columns of $\Gamma$ where
$1\leq r\leq q(q-1)-1$ and will prove that the columns in $R$  cover at least $r$
points. We assume that
$r=s + x$, where $s$ is the number of columns of $R$ from the parallel classes, and $x$, $0 \leq x \leq q-3$, is  the number of columns of $R$ which belong to the special class.
We use the same notations for $s$
as in the proof of Theorem~\ref{trm:td-cbc}, $s=iq+j$,
for $0\leq i\leq q-2$,
$0\leq j\leq q-1$. Similarly, $t\geq i$ is the maximum number of columns  which is contributed to $R$ by a parallel class.

We consider the same cases for $t$ as in the proof of Theorem~\ref{trm:td-cbc}, and present the details only when the proofs are different.

\textbf{Case $t\geq i+2$}. We assume here that $i+2\leq q-1$, otherwise $q$ columns of a parallel class cover $q(q-1)$ points.  Note that there is at least one additional parallel class which contributes at least $i$ columns (otherwise the total number of selected columns from the parallel classes is at most $i(q-1)$), therefore, $i+2$ columns from one class and $i$ columns from another class cover at least
 $(i+2)(q-1)+i-1\geq iq+(q-1)+(q-3)\geq iq+j+x=r$ points, by Lemma~\ref{lm:permutation}.

\textbf{Case $t=i+1$}. The only difference to the corresponding case of the proof of Theorem~\ref{trm:td-cbc} is when considering $t=i+1$ columns from a parallel class and $x$ columns from the special class, these columns cover at least $(i+1)(q-1)+x(q-1-i-1)=iq+(q-1)+x(q-i-2)-i$, then for $i\leq q-4$ and $i\leq x$ we are done.  The additional case $i=q-2$ with two parallel classes that contribute $q-1$ columns to $R$ and the case $i=q-3$ are identical to the proof of the corresponding cases in Theorem~\ref{trm:td-cbc}. The case $i=q-2$ with only one parallel class which contributes $q-1$ columns to $R$ corresponds to the case when $j\leq 1$, as in the proof of Theorem~\ref{trm:td-cbc}, and then two parallel classes that contribute $q-1$ and $q-2$ columns to $R$, respectively, together cover by Lemma~\ref{lm:permutation} at least  $(q-1)^2+q-3=(q-2)q+1+(q-3)\geq s+x$ points.

\textbf{Case $t=i$}. The only difference to the corresponding case of  Theorem~\ref{trm:td-cbc} is when considering $t=i$ columns from a parallel class and $x$ columns from the special class, these columns cover at least $i(q-1)+x(q-1-i)=iq+(q-1)+x(q-i-1)-i$, then for $i\leq q-3$ and $i\leq x$ we are done. Then the additional case is when $i=q-2$ (and $x\leq q-3$, by definition) which is identical to the proof of the corresponding case in Theorem~\ref{trm:td-cbc}.

Finally, similarly to the proof of Theorem~\ref{cor:uniform codes}, we have
 $$
 \left\lfloor\frac{(k-1)\binom{m}{c}}{\binom{k-1}{c}}\right\rfloor-n=q^2+q-2-(q^2+q-3)=1.
 $$

\end{proof}
 %%%%%%%%%%%%%%%%%%%%%%%%%%%%%%%%%%%%%%%%%%%%%%%%%%%%%%%%%%%%%%%%%%%%%%%%%%%%%%%%%%%%%%%%%

 Now we modify Construction IV to obtain an optimal uniform code,
with parameters that are different from the parameters of
the affine plane
based code of Construction II.

 \textbf{Construction V:}
 Let $\Gamma$ be the incidence matrix of a code, which is obtained by removing the first row (which corresponds to a server) and all the columns in the set $\{(i-1)q+1: 1\leq i\leq q\}$, (the columns which correspond to the items of the removed server) from the incidence matrix of $\cC_2(q)$.
  We denote the resulting uniform CBC by $\cC_3(q)$.

  \begin{example} The incidence matrix of the uniform code $\cC_3(4)$  is given by
  \end{example}

\begin{footnotesize}
$$\Gamma=\begin{array}{c}
     \left(\begin{tabular}{c|c|c|c|c}%\hline\hline
     1 0 0 &  1 0 0 &  1 0 0 &  1 0 0 & 1\\
     0 1 0 &  0 1 0 &  0 1 0 &  0 1 0 & 1\\
     0 0 1 &  0 0 1 &  0 0 1 &  0 0 1 & 1\\
     \hline
     0 0 0 &  1 0 0 &  0 0 1 &  0 1 0 &  0\\
     1 0 0 &  0 0 0 &  0 1 0 &  0 0 1 &  0\\
     0 1 0 &  0 0 1 &  1 0 0 &  0 0 0 &  0\\
     0 0 1 &  0 1 0 &  0 0 0 &  1 0 0 &  0\\
     \hline
     0 0 0 &  0 1 0 &  1 0 0 &  0 0 1 & 0 \\
     1 0 0 &  0 0 1 &  0 0 0 &  0 1 0 & 0 \\
     0 1 0 &  0 0 0 &  0 0 1 &  1 0 0 & 0 \\
     0 0 1 &  1 0 0 &  0 1 0 &  0 0 0 & 0 \\
         \end{tabular}
\right)
\end{array}$$
\end{footnotesize}

  \begin{theorem} The code $\cC_3(q)$ is an optimal $(q-1)$-uniform  $(q^2-3, (q-1)(q^2-3), q^2-q-1, q^2-q-1)$-CBC, attaining bound~(\ref{eq:UpBoundUniform}).
  \end{theorem}

 \begin{proof}
First, the parameters $n$, $N$, $m$, and  $c$ directly follow from the definition of the code.
Then, we need to prove that the parameter $k$ of Construction V is  the same as in Construction IV.
  Consider any  set $R$ of $r$ columns of $\Gamma$, $1\leq r\leq q^2-q-1$.  We expand every column in $R$ by adding a zero in the first position. These expanded columns are the columns of the incidence matrix of the uniform code $\cC_2(q)$, which cover $r$ points, not including the first point. Thus the original columns from $R$ cover $r$ points.

 Next we prove that the code has  the optimal  number of data items.
 Given that $m=k=q^2-q-1$ and $c=q-1$ we have
 $$\left\lfloor\frac{(k-1)\binom{k}{c}}{\binom{k-1}{c}}\right\rfloor
 =\left\lfloor\frac{(q^2-3)(q^2-2q)+q^2-3q+2}{q^2-2q}\right\rfloor=q^2-3=n.
 $$
 \end{proof}

\section*{Acknowledgment}
The authors thank the anonymous
referees for their valuable comments that helped to improve the
presentation of the paper.

%\bibliographystyle{abbrv}
%\bibliography{batch_bib}

\end{document}